\documentclass[11pt]{article}
\usepackage{textcomp}
\usepackage[leqno]{amsmath}
\usepackage[T1]{fontenc}
\usepackage{charter,eulervm}
\usepackage{setspace}
\usepackage{amsthm}
\usepackage{amssymb}
\usepackage{amsmath}
\usepackage[]{a4}

\usepackage{hyperref}
\usepackage{xcolor}
\definecolor{dark-red}{rgb}{0.4,0.15,0.15}
\definecolor{dark-blue}{rgb}{0.15,0.15,0.4}
\definecolor{medium-blue}{rgb}{0,0,0.5}
\hypersetup{
    colorlinks, linkcolor={dark-red},
    citecolor={dark-blue}, urlcolor={medium-blue}
}
\usepackage{enumerate}

\usepackage[round]{natbib}
\newtheorem{theorem}{Theorem}
\newtheorem{proposition}{Proposition}
\newtheorem{lemma}{Lemma}

\theoremstyle{definition}

\makeatletter

\def\lines@per@page{38} 
\def\prop@odd@margin{.5} 
\@settopoint\oddsidemargin \@settopoint\marginparwidth
\@settopoint\evensidemargin

\normalsize

\makeatother

\makeatletter 
\@ifundefined{DOTSB}{%
   
  \let\slimits@\displaylimits
  \let\ilimits@\nolimits }{}%
\newcommand{\genint}{}
\DeclareRobustCommand{\genint}[1]{\@optarg{\toks@{\@genint{#1}}\@geninta}[]}
\def\@geninta#1{\@optarg{\toks@\expandafter{\the\toks@{#1}}\@genintb}[]}
\def\@genintb#1{\@optarg{\the\toks@{#1}}[]}
\def\@genint#1#2#3#4#5{#1\ilimits@ 
  \@ifnotempty{#3}{_{#3}}%
  \@ifnotempty{#4}{^{#4}}%
  #5\@ifnotempty{#2}{\diffl{#2}}%
} 

\DeclareRobustCommand{\diffl}[1]{\mathinner{\mathrm d#1}}

\def\@optarg#1[#2]{%
  \def\@optarga[##1]{#1{##1}}%
  \begingroup
  \def\@optarga{\endgroup#1{#2}}%
  \futurelet\@let@token\@optargb }
\def\@optargb{\ifx[\@let@token\endgroup\fi \@optarga}

\makeatother

\DeclareMathOperator{\Int}{int}

\DeclareMathOperator{\con}{con}

\title{\Large\bfseries The Core in a Distributional Economy}

\author{
  Michael Greinecker\thanks{CEPS, ENS Paris-Saclay
  \href{mailto:michael.greinecker@ens-paris-saclay.fr}{michael.greinecker@ens-paris-saclay.fr}
  }~ and Konrad Podczeck\thanks{University of Vienna,
  \href{mailto:konrad.podczeck@univie.ac.at}{konrad.podczeck@univie.ac.at}}
}

\begin{document}
\maketitle

\begin{abstract} 
An economy, large or small, has traditionally been defined in terms of
an explicit set of agents and an assignment of characteristics to each
agent. But when individual agents are negligible, most economically
relevant properties of an economy can be defined in terms of the
distribution of characteristics alone. Agents need not be specified.

It has been frequently asserted that the distributional description of
an economy is too sparse for core analysis. Notions of coalitions and
blocking require the individualistic description of agents. This paper
shows that this is not so. The presence of blocking coalitions can be
directly identified in terms of distributions alone. Indeed, we give a
purely distributional proof of the classical core-equivalence theorem
that delivers the core-equivalence theorem for individualistic
economies as a corollary.

Our methods have applications outside of general equilibrium theory. They apply
to large matching markets and to analogs of the Shapley-value for atomless economies.
\end{abstract}

\section{Introduction}

Large economies have traditionally be defined as an assignment of
characteristics to every individual trader taken from a continuum.
However, competitive equilibria can be defined in terms of the distribution of
characteristics alone, as has been done by \citet*{hhk75} and
\citet{MR0434346}, and no specification of individual traders is
needed. The distributional description of an economy is more parsimonious,
allows us to study continuity properties of the equilibrium-correspondence,
and allows us to relate large economies to the limits of economies with finitely many traders.

So far, the core has eluded such a distributional formulation,
and a number of authors have written that such a distributional
formulation is not possible.\footnote{For example, \citet*{hhk75}
write that ``[T]he concepts of \lq coalition\rq{} and \lq to improve
upon\rq{} require the individualistic description of an economy as a
mapping which assigns to every individual agent his characteristics.''
and \citet{MR0395752} writes that ``for the equivalence theorem it is
indispensable to consider the economy in representation form.''} As we
show in this paper, it is. We characterize the role of blocking coalitions on a 
distributional implications and show that one could work with these distributional representations alone.

To understand the underlying economic issues, it helps to look at economies 
with finitely many traders and at the information that gets lost when looking at the 
distribution alone. Information is lost about the identity of individual traders and
about how many traders there are with specific characteristics, how many traders are 
represented by a single type. Without knowing the identity of a trader, we can not
track in two allocations whether the trader is better off in one rather than the other.
Without knowing how many traders are represented by a single type, we can neither
judge how competitive the economy is nor can we rely on the convexifying effect of large 
numbers.

These are not just problems one needs to address in order to understand the core, these problems
need to be addressed already to understand the simpler concept of Pareto efficiency. The problem of 
how to keep track of traders' identity has a simple solution: One simply looks at the whole joint distribution
of characteristics and the commodity bundles received in two allocations. This is enough to identify Pareto-rankings;
see Proposition \ref{Paretosimpel} below. The problem that we do not know how many traders a type represents becomes moot
if the underlying space of traders is sufficiently competitive. This can be achieved using a strengthened nonatomicity
assumption on the space of traders, the space must be superatomless; see Proposition \ref{Paretosat} below.

We also show what the usual Pareto-ordering means in a purely
distributional setting without looking at the whole joint distribution
of characteristics and the commodity bundles received in two allocations.
For short, let us refer to a measure on the
product of a space of trader's characteristics and the commodity space
as a distributional allocation, provided that its marginal on the
former space coincides with a given distribution of trader's
characteristics.
\citet{MR511487} have introduced an ordering on distributional
allocations in the spirit of first-order stochastic dominance that is
meant to capture the distributional content of the Pareto
ordering. They showed that if two allocations can be Pareto ordered,
then the induced distributional allocations satisfy their order.  Reif
and Wiesmeth also conjectured that a converse holds for some
individualistic representations of the distributional allocations. We show that this
conjecture is indeed correct in Theorem \ref{ParetoCoupling}. A celebrated theorem of
\citet{MR177430} constructs appropriate joint distributions for
stochastically ordered distributions. \citet{MR1711599} has pointed
out that the result holds also for stochastic pre-orderings, and this
is the main tool to prove Theorem \ref{ParetoCoupling}.

These results have analogs for the core with some
modifications. In particular, we can meaningfully define the core of a
distributional economy so that an allocation is in the classical core
if and only if its induced distribution is in the distributional
core. Importantly, even though we need a strengthened nonatomicity
condition on an underlying space of traders to represent conceivable
coalitions in the distribution exactly, which we do in Theorem \ref{coalitionexist},
the presence of a conceivable blocking coalition in the distribution implies the existence
of a blocking coalition even when the underlying space of traders is merely atomless; 
see Theorem \ref{approxcoalition}. Our tool for proving this is 
Lyapunov's theorem for Young measures, a result due to \citet{MR1798830}.

In order to show that our notion of a distributional core is
workable, we present a distributional core-equivalence theorem;
Theorem \ref{coreequiv}. We only assume that
preferences are continuous, irreflexive, and strictly monotone, and
that the aggregate endowment is strictly positive. No transitivity
assumption is needed.  Our proof makes no use of an ancillary
individualistic representation. One can obtain the general core-equivalence theorem
of  \citet{hildenbrand1982core} for individualistic economies with an atomless 
space of traders as a corollary. Our distributional methods are as powerful as
the individualistic approach.

The way we define blocking coalitions in a distributional economy is
closely related to various notions of generalized coalitions that have
been employed in the literature in equilibrium theory, value theory,
and the theory of stable matchings. We show that our approach allows us to give a natural
classical individualistic interpretation of the ``ideal sets'' of \citet{MR0378865},
the ``fuzzy coalitions'' of \citet{MR556865}, the ``convex set of
agents'' of \citet{MR873762}, and the ``subpopulations'' of
\citet*{MR1440419} and \citet*{che2013stable}. They can all be
interpreted as distributions induced by classical coalitions in a sufficiently rich space of
agents; the original set of agents can then be interpreted as a
space of types.

Throughout, we relate distributional allocations to individualistic allocations on a suitably rich (superatomless) underlying space of agents. \citet*{SUN2020105083} have argued that one should intepret distributional allocations as randomized allocations on an underlying non-rich space of agents. We think having a rich space of agents is economically more natural. Actually, \citet*{SUN2020105083} formulate some of their results in terms of ``rich Fubini extensions'' that require the space of agents to be superatomless. In that case, every distributional allocation admits an individualistic representation.    
\medskip

\noindent The remainder of the paper is organized as follows: We provide a distributional perspective on generalized coalitions and use it to motivate certain notions of nonatomicity in Section \ref{gencoal}. The economic environment is given in Section \ref{environ}. Section \ref{pareto} shows how to represent Pareto efficiency in distributional economies. Section \ref{core} does the same for the core. We use our characterization to provide a distributional core-equivalence in Section \ref{coequ}. All proofs can be found at the end in Section \ref{proofs}.

\section{Generalized Coalitions and Nonatomicity}\label{gencoal}

Before going into the specific context of finite-dimensional exchange economies, we put the problem into a broader context and explain the need for some heavy mathematical machinery.

Many authors have found at some point a need to generalize the
classical notion of a coalition as a measurable set of agents with
positive measure. We show that our distributional point of view allows
for a unified perspective on the existing notions of generalized
coalitions. Here are some examples of generalized coalitions coming up in economics:
\begin{enumerate}
\item Suppose you could order a continuum of players at random. For a given
set $S$ of players, the mass of players in $S$ who end up in the top
half under the random ordering should be one half of the mass of
players in $S$ almost surely. The set $S$ would be ``evenly spread''
under the random ordering, to use the language of
\citet{MR0378865}. An entirely satisfactory notion of such a random
ordering is impossible, as shown by Aumann and Shapley, and even the
concept of an ``evenly spread measurable set'' turns out to be empty
in their framework. However, they provide a convincing intuition of
the value of a nonatomic game with side-payments in terms of evenly
spread sets. In order to translate this intuition into mathematics,
Aumann and Shapley introduce \emph{ideal sets}, which are simply
measurable functions from the space of players to the unit
interval. If $g$ is such an ideal set and $i$ a player, $g(i)$ tells
us how much player $i$ belongs to the ideal set $g$. Ordinary
measurable sets are simply identified with their indicator functions.

\item With infinite dimensional commodity spaces, a non-atomic measure space
of agents may not be sufficient for perfect competition, markets might
fail to be ``thick.'' An interesting test for thick markets was
introduced by \citet{MR873762}. Gretzky and Ostroy extend the
underlying nonatomic probability space of agents to a \emph{convex
space of agents} consisting of simple functions with values in the
unit interval. Markets are then considered thick when the space of
allocations remains essentially unchanged when going from the usual
space of agents to the convex space of agents. The latter is meant to
express perfect divisibility on the level of coalitions and can be
seen as a continuum replica version of the original economy.

\item A distributional approach is taken by \citet*{MR1440419} and
\citet*{che2013stable}. In both papers, there is a compact metric
space of characteristics $K$, and the population of agents is given by
a nonnegative Borel measure $\mu$ on $K$. A \emph{subpopulation} in
these papers is then a nonnegative Borel measure $\nu$ on $K$ such
that $\nu(A)\leq\mu(A)$ for every Borel set $A\subseteq K$. The actual
space of agents is not formalized. By the Radon-Nikodym theorem,
subpopulations can be identified with (equivalence classes of)
measurable functions with values in the unit interval.
\end{enumerate}

Let $(X,\mathcal{X},\lambda)$ be a probability 
space. A \emph{generalized coalition} is a $\lambda$-equivalence class
of measurable functions $c:X\to [0,1]$. A \emph{coalition} is a
$\lambda$-equivalence class of measurable functions
$c:X\to\{0,1\}$. For \citet{MR0378865}, generalized coalitions are
\emph{ideal sets}. \citet{MR873762} added the subclass of generalized
coalitions corresponding to simple functions to ordinary
coalitions, to enrich the \emph{indivisible space of agents} to the
\emph{convex space of agents}. \citet{MR556865} generalized
coalitions to \emph{fuzzy coalitions}, a terminology inspired by
the corresponding notion of a \emph{fuzzy set} of \citet{MR0219427}.

We show that generalized coalitions on a probability space $(X,\mathcal{X},\lambda)$ can be viewed as distributions induced by individualistic coalitions when one interprets elements of $X$ as types of individuals instead of as proper individuals.\footnote{A similar distributional interpretation of generalized coalitions
has already been given by \citet{MR1300524} for the special case of a finite space $X$.} The setting we look at is that there is an underlying probability space of agents  $(T,\Sigma,\nu)$ and a measurable function $\phi:T\to X$ with distribution $\lambda$. Under this distributional point of view, generalized coalitions are naturally treated as subpopulations and we do so in what follows. We show that under a suitable nonatomicity hypothesis, every subpopulation of $(X,\mathcal{X},\lambda)$ corresponds to a proper coalition in $\Sigma$.  

The relevant nonatomicity notion goes back to \citet{MR0006595} and a lemma from that paper (Lemma \ref{maharam} in Section \ref{proofs}) will be our main tool. Let $(T,\Sigma,\nu)$ be a finite measure space and $\Sigma'\subseteq\Sigma$ a sub-$\sigma$-algebra. We say that $\nu$ is \emph{relatively atomless} over $\Sigma'$ if for every $A\in\Sigma$ such that $\nu(A)>0$, there exists $B\in\Sigma$ such that $B\subseteq A$ and $\nu(B\Delta C)>0$ for all $C\in\Sigma'$. The relation to the classical notion of being atomless is a follows: The measure $\nu$ is atomless if and only if for every $A\in\Sigma$, $\nu$ is relatively atomless over $\{T,A,T\setminus A,\emptyset\}$.  

For a function $f$ into a measurable space, we write $\sigma(f)$ for the $\sigma$-algebra generated by $f$, the coarsest $\sigma$-algebra on the domain of $f$ that makes $f$ measurable. The notion of being relatively atomless makes it possible to represent subpopulations of a distribution induced by a funtion from the space of agents to a space of characteristics as actual coalitions in the space of agents. 

\begin{theorem}\label{gencoalition}Let $(T,\Sigma,\nu)$ be a probability space, $(X,\mathcal{X})$ a measurable space and $\phi:T\to X$ a measurable function. Suppose $\nu$ is relatively atomless over $\sigma(\phi)$. If $\mu$ is a measure on $(X,\mathcal{X})$ such that $\mu(A)\leq\nu\circ\phi^{-1}(A)$ for all $A\in\mathcal{X}$, then there exists a set $C\in\Sigma$ such that \[\mu(A)=\nu\big(\phi^{-1}(A)\cap C\big)\] for all $A\in\mathcal{X}$.
\end{theorem} 

The assumption that $\nu$ is relatively atomless over $\sigma(\phi)$ cannot be weakened; the conclusion of Theorem \ref{gencoalition} is actually equivalent to $\nu$ being relatively atomless over $\sigma(\phi)$. The converse to Theorem \ref{gencoalition} is the following proposition.

\begin{proposition}\label{relatomless} Let $(T,\Sigma,\nu)$ be a probability space, $(X,\mathcal{X})$ a measurable space and $\phi:T\to X$ a measurable function.  If for every measure $\mu$ on $(X,\mathcal{X})$ such that $\mu(A)\leq\nu\circ\phi^{-1}(A)$ for all $A\in\mathcal{X}$ there exists a set $C_\mu\in\Sigma$ such that \[\mu(A)=\nu\big(\phi^{-1}(A)\cap C_\mu\big)\] for all $A\in\mathcal{X}$, then $\nu$ is relatively atomless over $\sigma(\phi)$.
\end{proposition}

Theorem \ref{gencoalition} is stated in terms of a single function $\phi$ from agents to characteristics. However, to study exchange economies from a distributional perspective, we want to be able to have a similar result that applies at once to all coalitions. So we need a probability space that is relative atomless with respect to every $\sigma$-algebra generated by an allocation. In our context, this means that we need a probability space that is relatively atomless over any countably generated sub-$\sigma$-algebra. These are the superatomless probability spaces. The notion of a superatomless measure space, as introduced by \citet{integ4}, is most naturally expressed in the
language of measure algebras.\footnote{Equivalent notions have been introduced in various fields of mathematics. Footnote 4 of \citet{MR2959114} provides pointers to various equivalent notions; one of them (saturation) will be used below.} Measure algebras are the result of
identifying measurable sets whose symmetric difference has measure
zero. By identifying measurable sets with their indicator functions,
this is a special case of identifying measurable functions that agree
almost everywhere. We use this identification to formulate
superatomlessness in more familiar terms.\footnote{The equivalence to
the original definition is proven in \citet{integ4}.}  Let
$(T,\Sigma,\nu)$ be a nontrivial finite measure space and $L_1(\nu)$ be the
corresponding space of equivalence classes of integrable functions
endowed with the usual $L_1$-norm $\|\cdot\|_1$. For each
$A\in\Sigma$, let $L_1(A,\nu)$ be the subspace of $L_1(\nu)$
consisting of elements vanishing outside of $A$. Now $(T,\Sigma,\nu)$
or $\nu$ is \emph{superatomless} if $L_1(A,\nu)$ is non-separable for
all $A\in\Sigma$ with $\nu(A)>0$.\footnote{A wide variety of
conditions equivalent to being superatomless can be found in the
literature. Footnote 4 in \citet{MR2959114} lists most of the known
equivalent conditions.} The following is implicitly well-known; we supply a proof for the 
reader's convenience.
\begin{proposition}\label{superrelatomless} A nontrivial finite measure space $(T,\Sigma,\nu)$ is superatomless if
and only if it is relatively atomless over any countably generated sub-$\sigma$-algebra $\Sigma'\subseteq\Sigma$.
\end{proposition}

There is a very useful property introduced by  \citet{hooverkeisler}, who termed it \emph{saturation}, that characterizes superatomless probability spaces: 

\begin{trivlist}
\item \emph{A measure space $(T,\Sigma,\nu)$ is superatomless if and
only if whenever $X$ and $Y$ are Polish spaces and $\kappa$ is a Borel
measure on $X\times Y$ such that its marginal on $X$ equals the
distribution of some measurable function $f: T\to X$, there is a
measurable function $g: T\to Y$ such that $\kappa$ is the distribution
of the function $t\mapsto\big(f(t),g(t)\big)$.}\footnote{The equivalence follows from
 \citet[Corollary 4.5]{hooverkeisler}.}
\end{trivlist}

Widely used probability spaces such as the unit interval with the Borel sets and Lebesgue measure are not saturated. That this might
cause problems when representing individualistic economies with such spaces as agents in distributional form has already been observed in \citet[Section 2.16.]{MR0395752}.

\section{Economic Environment}\label{environ}

The commodity space is $\mathbb{R}^l$ and each trader has the same
consumption set $E=X=\mathbb{R}^l_+$. These spaces are endowed with their natural Euclidean topologies. Every topological space is viewed as being endowed with its Borel $\sigma$-algebra. We also use notation like $X'$
to distinguish different occurrences of $X$ in a product space. We let
$\mathcal{P}$ be the set of irreflexive relations on $X$ with a
relatively open graph, endowed with the coarsest topology such that
the set
\[\{(\succ,x,y)\in\mathcal{P}\times X\times X\mid x\succ y\}\] is open in
the product topology. Under this topology, $\mathcal{P}$ is compact
and metrizable. The subspace $\mathcal{P}^*$ of $\mathcal{P}$
consisting of asymmetric and negatively transitive relations is a
$G_\delta$-subset of $\mathcal{P}$ and,  therefore, a Polish subspace by
Alexandroff's lemma. Note that an asymmetric and negatively transitive relation is simply the assymmetric part of a complete and transitive relation---the usual assumption on preference relations. Similarly, the subspace
$\mathcal{P}_{\textnormal{mo}}$ of strictly monotone preferences, those 
preferences for which $x>y$ implies $x\succ y$, is a $G_\delta$-subset of $\mathcal{P}$, 
therefore also a Polish subspace. Also, we let
$\mathcal{P}_{\textnormal{mo}}^*=\mathcal{P}^*\cap
\mathcal{P}_{\textnormal{mo}}$. These results all follow from
\citet[1.2]{MR0389160}.

An \emph{individualistic economy} consists of a probability space
$(T,\Sigma,\nu)$ of traders and a measurable function
$\mathcal{E}:T\to\mathcal{P}\times E$ such that the second coordinate
function is $\nu$-integrable. We let $c$ be the first coordinate
function and $\eta$ be the second coordinate function of
$\mathcal{E}$. We also write $\succ_t=c(t)$ and $e_t=\eta(t)$. We call
$\nu\circ\mathcal{E}^{-1}$ the \emph{distribution} of the economy. An
\emph{individualistic allocation} $f$ for the economy is a measurable
function $f:T\to X$ such that
\[\int f~\textnormal{d}\nu=\int \eta~\textnormal{d}\nu.\]

We also introduce distributional versions of these concepts. A
\emph{distributional economy} is a Borel probability measure $\mu$ on
$\mathcal{P}\times E$ such that the canonical projection
$\pi_E:\mathcal{P}\times E\to E$ is $\mu$-integrable. A
\emph{distributional allocation} for the distributional economy $\mu$
is a Borel Probability measure $\alpha$ on $\mathcal{P\times X}$ such
that $\alpha$ and $\mu$ have the same $\mathcal{P}$-marginal, and such
that
\[\int \pi_{X}~\textnormal{d}\alpha=\int \pi_E~\textnormal{d}\mu,\]
with $\pi_X$ being the projection onto $X$. We also define a \emph{
reallocation} for the distributional economy $\mu$ to be a Borel
probability measure on $\mathcal{P}\times E\times X$ whose
$\mathcal{P}\times E$-marginal is $\mu$ and whose
$\mathcal{P}\times X$-marginal is a distributional allocation for the
distributional allocation $\mu$.

\section{Pareto Efficiency}\label{pareto}

Fix an individualistic economy and let $f$ and $f'$ be individualistic
allocations. We write $f\succeq_P f'$ if
the set of $t$ such that $f'(t)\succ_t f(t)$ has $\nu$-measure
zero. The asymmetric part
$\succ_P$ of $\succeq_P$ represents the usual relation of \emph{Pareto
dominance} and a $\succeq_P$-maximal allocation is \emph{Pareto
efficient}. It is essential for a treatment of Pareto efficiency that
we allow for a measure zero exceptional set because feasibility is
defined by integration. The following is straightforward.
\begin{proposition}\label{Paretosimpel}Let $f$ and $f'$ be
individualistic allocations. Then $f\succeq_P f'$ if and only if the
set
\[\big\{(\succ,x,x')\mid x'\succ x\big\}\]
has $\nu\circ (c,f,f')^{-1}$-measure zero.
\end{proposition}
The following is a direct consequence of the saturation property of
superatomless probability spaces.
\begin{proposition}\label{Paretosat} Let $\nu$ be superatomless. Let
$f$ be an individualistic allocation and $\mu$ be any probability
measure on the set $\mathcal{P}\times X\times X'$, with
$\mathcal{P}\times X$-marginal $\nu\circ (c,f)^{-1}$, such that
\[\int \pi_{X'}~\textnormal{d}\mu=\int \pi_X~\textnormal{d}\mu,\]
where $\pi_{X}$ and $\pi_{X'}$ are the projections onto $X$ and $X'$,
respectively. Then there exists an individualistic allocation $f'$
such that $\mu=\nu\circ(c,f,f')^{-1}$.
\end{proposition}
Together, the last two propositions show that for a superatomless
space of agents, Pareto efficiency can be fully characterized in terms
of the induced joint distributions; the joint distribution with a Pareto
dominating distributional allocation can be exactly represented by an
individualistic allocation. If $(T,\Sigma,\nu)$ is only assumed to be
atomless, this is in general not possible. Nevertheless, Pareto
efficiency can be defined in distributional terms as the following
theorem shows.

\begin{theorem}\label{approxPareto}Assume that $(T,\Sigma,\nu)$ is
atomless. Let $f$ be an individualistic allocation and $\mu$ a Borel
probability measure on $\mathcal{P}\times E\times X\times X'$, with
$\mathcal{P}\times E\times X$-marginal $\nu\circ(c,\eta,f)^{-1}$, such
that the set
\[\big\{(\succ,e,x,x')\mid x\succ x'\big\}\]
has $\mu$-measure zero and such that the set
\[\big\{(\succ,e,x,x')\mid x'\succ x\big\}\]
has positive $\mu$-measure.  Then there exists an individualistic
allocation $f'$ such that $f'\succ_P f$.
\end{theorem}

It is crucial for the argument that one has information of the joint
distribution of $f$ and $f'$. Indeed, if $\mu$ is a measure on
$\mathcal{P}\times X\times X'$ that puts zero probability on the set
$\big\{(\succ,x,x')\mid x\succ x'\big\}$, and if $f$ and $f'$ are
allocations such that $\nu\circ (c,f)^{-1}$ equals the
$\mathcal{P}\times X$-marginal of $\mu$ and $\nu\circ (c,f')^{-1}$
equals the $\mathcal{P}\times X'$-marginal of $\mu$, then we still do
not know whether $f'\succeq_P f$. However, this is enough to show that
there is another allocation $f''$ such that
$\nu\circ (c,f')^{-1}=\nu\circ (c,f'')^{-1}$ and $f''\succeq_P
f'$. This, in turn, is enough to determine whether $f$ is
$\succeq_P$-maximal, that is, Pareto efficient. Pareto efficiency can
be characterized in distributional terms.\bigskip

However, it would be desirable to make such comparisons without
specifying a possible joint distribution. The following formulation in
the spirit of stochastic dominance due to \citet{MR511487} will allow
us to do just that. Let $\alpha$ and $\alpha'$ be distributional
allocations for the distributional economy $\mu$. We write
$x\succeq y$ for $y\nsucc x$\footnote{For $\succ\in\mathcal{P}^*$, $\succeq$ is a complete and transitive relation.} and let
$B(\succeq,x)=\{y\in Y\mid y\succeq x\}$. We write
$\alpha\succeq_D\alpha'$ if for some (and hence all) regular
conditional distributions $g_\alpha:\mathcal{P}\to\Delta(X)$ and
$g_{\alpha'}:\mathcal{P}\to\Delta(X)$ of $\alpha$ and $\alpha'$,
respectively, and each $x\in X$ one has
\[g_\alpha(\succ)\Big(B(\succeq,x)\Big)\ge
g_{\alpha'}(\succ)\Big(B(\succeq,x)\Big)\] for
$\mu_\mathcal{P}$-almost all $\succ$, where $\mu_\mathcal{P}$ is the
$\mathcal{P}$-marginal of $\mu$. Reif and Wiesmeth have shown that the
usual Pareto ordering of individualistic allocations in an
individualistic economy induces the ordering $\succeq_D$ on the
corresponding distributional allocations and conjectured that the
converse holds for some individualistic representation of these
distributional allocations when preferences are in $\mathcal{P}^*$. We show that
this conjecture is indeed true in the next theorem.

The interpretation of the ordering is as follows: Look at all agents
with preference exactly $\succ$. Then for every $x$, ``at least as
many'' agents with this preference relation get something at least as
good as $x$ under $\mu$ than under $\mu'$. By reordering agents, we
could make sure that every agent is at least as well off under $\mu$
than under $\mu'$.

\begin{theorem}\label{ParetoCoupling}Let $\mu$ be a distributional
economy such that $\mu(\mathcal{P}^*\times E)=1$. Let $\alpha$ and $\alpha'$ be two
distributional allocations for the economy $\mu$ such that
$\alpha\succeq_D\alpha'$. Then there exists a measure $\lambda$ on
$\mathcal{P}\times X\times X'$ such that the
$\mathcal{P}\times X$-marginal of $\lambda$ is $\alpha$, the
$\mathcal{P}\times X'$-marginal of $\lambda$ is $\alpha'$ and
\[\lambda\Big(\big\{(\succeq,x,x')\mid x'\succ x\big\}\Big)=0.\]
\end{theorem}

\section{The Core}\label{core}

We next go to the slightly more involved problem of analyzing blocking
coalitions from a distributional point of view. Consider an
individualistic economy with space of traders $(T,\Sigma,\nu)$ and
preferences and endowments given by
$\mathcal{E}:T\to\mathcal{P}\times E$. An \emph{individualistic
coalition} is simply a measurable set $C\in\Sigma$ such that
$\nu(C)>0$. We say that the individualistic coalition $C$
\emph{blocks} the individualistic allocation $f$ if there is an
individualistic allocation $f'$ such that $f'(t)\succ_t f(t)$ for
$\nu$-almost all $t\in C$ and
\[\int_C f'~\textnormal{d}\nu=\int_C \eta~\textnormal{d}\nu.\]
In that case, we say that $C$ \emph{blocks} $f$ \emph{by} $f'$.  We
can still work with distributions, but have to look at the
distribution of the measure when restricted to $C$. The distribution
of an individualistic allocation $f$ restricted to an individualistic
coalition $C$ assigns to each Borel subset $B$ of
$\mathcal{P}\times E\times X$ the measure
$\nu\big((c,\eta,f)^{-1}(B)\cap C\big)$. This restricted distribution is
not a probability measure unless $\nu(C)=1$. The following is, again,
straightforward.
\begin{proposition}\label{coresimple}Let $f$ and $f'$ be
individualistic allocations and $C$ an individualistic
coalition. Let \[\rho(A)=\nu\big((c,\eta,f,f')^{-1}(A)\cap C\big).\]
Then $C$ blocks $f$ by $f'$ if and only if the set
\[\big\{(\succ,e,x,x')\mid x'\nsucc x\big\}\]
has $\rho$-measure zero and
\[\int \pi_{X'}~\textnormal{d}\rho=\int\pi_E~\textnormal{d}\rho.\]
\end{proposition}
Much like distributional allocations, we can exactly represent
``distributional blocking coalitions'' when $(T,\Sigma,\nu)$ is
superatomless.

\begin{theorem}\label{coalitionexist}Assume that $(T,\Sigma,\nu)$ is
superatomless. Let $f$ be an individualistic allocation and $\mu$ a
Borel probability measure on $\mathcal{P}\times E\times X\times X'$
such that the set
\[\big\{(\succ,e,x,x')\mid x'\nsucc x\big\}\]
has $\mu$-measure zero, such that the
$\mathcal{P}\times E\times X$-marginal of $\mu$ is setwise smaller
than $\nu\circ(c,\eta,f)^{-1}$, and such that
\[\int \pi_{X'}~\textnormal{d}\mu=\int\pi_E~\textnormal{d}\mu.\]
Then there exists an individualistic coalition $C$ and an
individualistic allocation $f'$ such that $C$ blocks $f$ by $f'$ and
for every Borel subsets $A$ of $\mathcal{P}\times E\times X\times X'$
one has $\mu(A)=\nu\Big((c,\eta,f,f')^{-1}(A)\cap C\Big)$.
\end{theorem}
 
This result also shows that for $\nu$ superatomless, every coalition
that might exist in the distribution actually exists as a proper individualistic
coalition. If $\nu$ is only assumed to be atomless, the existence of
such a $\mu$ still guarantees that some coalition can block the
allocation $f$.

\begin{theorem}\label{approxcoalition}Assume that $(T,\Sigma,\nu)$ is
atomless. Let $f$ be an individualistic allocation and $\mu$ a Borel
probability measure on $\mathcal{P}\times E\times X\times X'$ such
that the set
\[\big\{(\succ,e,x,x')\mid x'\nsucc x\big\}\]
has $\mu$-measure zero and such that the
$\mathcal{P}\times E\times X$-marginal of $\mu$ is setwise smaller
than $\nu\circ(c,\eta,f)^{-1}$ and such that
\[\int \pi_{X'}~\textnormal{d}\mu=\int\pi_E~\textnormal{d}\mu.\]
Then there exists an individualistic coalition $C$ and an
individualistic allocation $f'$ such that $C$ blocks $f$ by $f'$.
\end{theorem}
 
Just as in the case of Pareto optimality, it is possible to compare
coalitions in distributional terms without an {\it a priori}
specification of a joint distribution. For coalitional blocking, we
need to make somewhat stronger assumptions. We do not just want to
have $f'(t)\succeq_t f(t)$ for almost all $t\in C$, we want
$f'(t)\succ_t f(t)$ for almost all $t\in C$. To do so, we take
preferences to be in $\mathcal{P}_{\textnormal{mo}}^*$. Then, if we
can show that $f'(t)\succeq_t f(t)$ for almost all $t\in C$ but not
$f(t)\succeq_t f'(t)$ for almost all $t\in C$, then there must be a
positive measure set $E\subseteq C$ such that $f'(t)\succ_t f(t)$ for
almost all $t\in E$ and, by a standard argument using strict
monotonicity, one can redistribute then between the members of $C$ so
that almost everyone is better off than under $f$. The other new
complication comes from the fact that we are not just considering the
grand coalition, we have to keep track of the total endowment of a
coalition. To make the argument more transparent, we look, restricted
to a distributional counterpart of a coalition, at the original
allocation with values in $X$, an allocation that makes some agents
better off and almost nobody worse off with values in $X'$, and an
allocation with values in $X''$ that makes almost everyone in the
coalition strictly better off.

\begin{theorem}\label{Corecoupling}Suppose that $\mu$ and $\mu'$ are
Borel probability measures (we just renormalize coalitions) on
$\mathcal{P}_{\textnormal{mo}}^*\times E\times X$, with the same
$\mathcal{P}_{\textnormal{mo}}^*\times E$-marginal, such that
\[\int \pi_{X'}~\textnormal{d}\mu'=\int
\pi_{E}~\textnormal{d}\mu',\] and such that for regular conditional
probabilities
\[d:\mathcal{P}_{\textnormal{mo}}^*\times E\to\Delta(X)\] and
\[d':\mathcal{P}_{\textnormal{mo}}^*\times E\to\Delta(X')\] of $\mu$
and $\mu'$, respectively, we have
\[d(\succ,e)\Big(B(\succeq,x)\Big)\ge
d'(\succ,e)\Big(B(\succeq,x)\Big)\] for almost all $(\succ,e)$ with
respect to the corresponding marginal of $\mu$ but not vice versa.
Then there exists a measure $\lambda$ on
$\mathcal{P}_{\textnormal{mo}}^*\times E\times X\times X'\times X''$
such that
\[\lambda\Big(\big\{(\succ,e,x,x',x'')\mid x\succ x'\textnormal{ or
}x'\succeq x''\big\}\Big)=0,\] such that the
$\mathcal{P}_{\textnormal{mo}}^*\times E\times X$-marginal of
$\lambda$ coincides with $\mu$, such that the
$\mathcal{P}_{\textnormal{mo}}^*\times E\times X'$-marginal of
$\lambda$ coincides with $\mu'$ and such that
\[\int \pi_{X''}~\textnormal{d}\lambda=\int
\pi_{E}~\textnormal{d}\lambda.\]
\end{theorem}

\section{Core-Equivalence}\label{coequ}

So far, we have shown that one can meaningfully define the relevant
concepts of Pareto efficiency and the core in distributional terms. It
is less clear that the distributional notions are operational, that one
can actually use them without specifying an ancillary probability space of
individual agents. We show that these notions are workable by using
them to prove a fairly general core-equivalence theorem. For this, we
need to introduce some new economic terms.

A reallocation $\tau$ for the distributional economy $\mu$ is
\emph{Walrasian} if there exists $p\in\mathbb{R}^l$ such that
\[\tau\Big(\big\{(\succ,e,x)\mid p\cdot x\leq p\cdot e\textnormal{ and
}x'\succ x\textnormal{ implies }p\cdot x'>p\cdot e\textnormal{ for all }x'\in
X\big\}\Big)=1.\] The set involved in this definition is easily seen
to be measurable; just note that the second condition needs only to be
checked for $x'$ in a countable dense subset of $X$. A reallocation
$\tau$ is a \emph{core reallocation} if there exists no measure $\mu$
on $\mathcal{P}\times E\times X\times X'$ such that the set
\[\big\{(\succ,e,x,x')\mid x'\nsucc x\big\}\]
has $\mu$-measure zero and such that the
$\mathcal{P}\times E\times X$-marginal of $\mu$ is setwise smaller
than $\tau$ and such that
\[\int \pi_{X'}~\textnormal{d}\mu=\int\pi_E~\textnormal{d}\mu.\]

When preferences are monotone and the aggregate endowment strictly positive, there is no difference between Walrasian relallocations and core reallocations.

\begin{theorem}\label{coreequiv}For a distributional economy $\mu$
such that $\mu(\mathcal{P}_{\textnormal{mo}}\times E)=1$ and
$\int\pi_E~d\mu\gg 0$, a reallocation is Walrasian if and only if it
is a core reallocation.
\end{theorem}

By combining Theorem \ref{coreequiv} with Proposition
\ref{coresimple}, Theorem \ref{approxcoalition}, and the easily checked fact
that an individualistic allocation is Walrasian if and only if
the induced reallocation is Walrasian, one has an alternative proof of
the very general core-equivalence theorem for an atomless economy of
\citet{hildenbrand1982core}.  The original core-equivalence theorem of
\citet{au64} does not assume preferences to be irreflexive, and, though
Aumann makes the same continuity assumption, his proof only requires
all $\succ$-better sets to be open. The only role the stronger
assumptions play in our proof and the proof of
\citet{hildenbrand1982core} is in guaranteeing that there is a
well-defined measurable space of preferences. Aumann does not define
an economy as a measurable function into a measurable space of
characteristics and can, therefore, directly specify the minimal
continuity and measurability assumptions needed for the proof to go
through; see the discussion in \citet{MR3518591}.

Our proof of Theorem 5 does not use any ancilliary representation of
the distributional economy as an individualistic economy; the
distributional point of view is entirely appropriate for core
analysis.

\section{Proofs}\label{proofs}

The following lemma, translated from the abstract setting of measure algebras to the setting of measure spaces, lies at the heart of the proof of Maharam's celebrated characterization result for measure algebras in \citet{MR0006595}. For an exceptionally clear proof, see \citet[Lemma 3.2]{MR991611}.

\begin{lemma}\label{maharam} Let $(T,\Sigma,\nu)$ be a finite measure space, $\Sigma'\subseteq \Sigma$ a sub-$\sigma$-algebra such that $\nu$ is relatively atomless over $\Sigma'$, and let $\mu$ be a measure on $(T,\Sigma')$ such that $\mu(A)\leq\nu(A)$ for all $A\in\Sigma'$. Then there exists a set $C\in\Sigma$ such that $\mu(A)=\nu(A\cap C)$ for all $A\in\Sigma'$.
\end{lemma}

\begin{proof}[Proof of Theorem \ref{gencoalition}] Clearly, $\mu$ is absolutely continuous with respect to $\nu\circ\phi^{-1}$ and has a Radon-Nikodym derivative $g:X\to [0,1]$ with respect to $\nu\circ\phi^{-1}$. Let $\kappa$ be the measure on $\Sigma$ that has Radon-Nikodym derivative $g\circ\phi$ with respect to $\nu$. Clearly, $\kappa(B)\leq\nu(B)$ for all $B\in\sigma(\phi)$. By the Lemma \ref{maharam}, there exists $C\in\Sigma$ such that $\kappa(B)=\nu(B\cap C)$ for all$B\in\sigma(\phi)$. Let $A\in\mathcal{X}$,
\[\mu(A)=\int_A g~d\nu\circ\phi^{-1}=\int g1_A~d\nu\circ\phi^{-1}=\int g1_A\circ\phi~d\nu\]
\[=\int (1_A\circ\phi)(g\circ\phi)~d\nu=\int 1_A\circ\phi~d\kappa=\int 1_{\phi^{-1}(A)}~d\kappa=\kappa\circ\phi^{-1}(A)\]
\[=\nu(C\cap\phi^{-1}(A)).\]
\end{proof}

\begin{proof}[Proof of Proposition \ref{relatomless}] Suppose for the sake of contradiction that there exists $A\in\Sigma$ with $\nu(A)>0$ such that for every $B\in\Sigma$ with $B\subseteq A$, there exists $C\in\sigma(\phi)$ such that $\nu\big(B\Delta C\big)=0$. In particular, there must exist $G\in\sigma(\phi)$ such that $\nu\big(A\Delta G\big)=0$. Let $\mu=1/2~\nu\circ\phi^{-1}$. It follows that $\nu(C\cap C_\mu)=1/2\nu(C)$ for all $C\in\sigma(\phi)$. Therefore, $\nu(A\cap C_\mu)=\nu(G\cap C_\mu)>0$. Since $A\cap C_\mu\in\Sigma$ and $A\cap C_\mu\subseteq A$, there is an $F\in\sigma(\phi)$ such that $\nu\big((A\cap C_\mu)\Delta F\big)=0$. But then \[1/2\nu(F)=\nu(F\cap C_\mu)=\nu\big((A\cap C_\mu)\cap C_\mu\big)=\nu\big(A\cap C_\mu\big)=\nu(F)>0,\]
which is absurd.
\end{proof}

\begin{proof}[Proof of Proposition \ref{superrelatomless}] 
Let first $\nu$ be relatively atomless over every countably generated sub-$\sigma$-algebra, $A\in\Sigma$ with $\nu(A)>0$, and let $S\subseteq L_1(A,\nu)$ be countable, and let $L$ be the smallest closed subspace of $L_1(A,\nu)$ containing $S$. Choose a representative of every element of $S$ and let $\Sigma'$ be the sub-$\sigma$-algebra of $\Sigma$ generated by these representatives. Then $\Sigma'$ is countably generated. By assumption, there exists  $B\in\Sigma$ with $B\subseteq A$ such that $\nu(B\Delta C)>0$ for all $C\in\Sigma'$. Therefore, the equivalence class of $1_B$ is not in $L$. Since $S$ was an arbitrary countable subset of $L_1(A,\nu)$,  $L_1(A,\nu)$ is not separable. Since this holds  for all $A\in\Sigma$ with $\nu(A)>0$, $\nu$ is superatomless.

For the other direction, assume that $\nu$ is superatomless, that $L_1(A,\nu)$ is not separable for any $A\in\Sigma$ with $\nu(A)>0$. Let $\Sigma'\subseteq\Sigma$ be countably generated and $A\in\Sigma$ with $\nu(A)>0$. Let $L$ be the subspace of $L_1(A,\nu)$ in which every element is an equivalence class of a $\Sigma'$-measurable function multiplied by $1_A$. Since $L_1$-space corresponding to countably generated $\sigma$-algebras are separable (simple functions with rational values that are masurable with respect to a countable generating algebra form a countable dense subset), there is an element of $L_1(A,\nu)$ that is not in $L$. Approximating a representative of this element by simple functions supported on $A$, we see that there must be some $B\subseteq A$ such that the equivalence class of $1_B$ is not in $L$. So for all $C\in\Sigma'$, $\int|1_B-1_C|~\mathrm d\nu=\nu(B\Delta C)>0$. So $\nu$ is relatively atomless over every countably generated sub-$\sigma$-algebra $\Sigma'\subseteq\Sigma$.
\end{proof}

\begin{proof}[Proof of Theorem \ref{approxPareto}]
Let $\kappa_\mu:\mathcal{P}\times E\times X\to \Delta(X')$ be a
disintegration of $\mu$. Define $\kappa:T\to\Delta(X')$ by
$\kappa=\kappa_{\mu}\circ (c,\eta,f)$. Let
$g_0:T\times X'\to\mathbb{R}$ be given by $g_0(t,x)=1$ if
$f(t)\succ_t x$ and $0$ otherwise. Let $g_1:T\times X'\to\mathbb{R}$
be given by $g_1(t,x)=1$ if $x\succ_t f(f)$. Finally, let
$g_2:T\times X'\to\mathbb{R}$ be given by $g_2(t,x)=x-e_t$. The
integrand $g_0$ evaluated at $\kappa$ is $0$, the integrand $g_1$
evaluated at $\kappa$ is positive, and the integrand $g_2$ evaluated
at $\kappa$ is $0$ By Lyapunov's theorem for Young measures,
\citet[Theorem 5.10]{MR1798830}, there is a function $f':T\to X'_{\o}$
such that
\[\int g_0\big(t,f'(t)\big)~\mathrm d\nu=\int \int g_0(t,x)~\mathrm
d\kappa(t)~\mathrm d\nu=0,\]
\[\int g_1\big(t,f'(t)\big)~\mathrm d\nu=\int \int g_1(t,x)~\mathrm
d\kappa(t)~\mathrm d\nu=\mu\Big(\big\{(\succ,e,x,x')\mid  x'\succ
x\big\}\Big)>0\] and
\[\int g_2\big(t,f'(t)\big)~\mathrm d\nu=\int \int g_2(t,x)~\mathrm
d\kappa(t)~\mathrm d\nu=0.\] It follows from the last condition that
$f'$ is an individualistic allocation, from the first condition that
$f'\succeq_P f$, and from the second condition that, indeed,
$f'\succ_P f$.
\end{proof}

The following lemma is stated in \citet*{MR0494447} without proof; we
supply a proof for the readers' convenience.

\begin{lemma}\label{monapprox}Let $(M,\mathcal{M})$ be a measurable
space, $\succcurlyeq$ a preorder on $M$ such that all
$\succcurlyeq$-lower sections are measurable, and $f:M\to [0,1]$ a
$\succcurlyeq$-nondecreasing function. Then there exists a sequence
$\langle f_n\rangle$ of simple functions $f_n:M\to [0,1]$ that
converges pointwise to $f$ such that each $f_n$ is a positive linear
combination of indicator functions of sets of the form
$B(\succcurlyeq,x)$.\end{lemma}
\begin{proof}Let $f_0$ be the constant function with value $0$. For
$n\geq 1$, let
\[K_n=\{k\in\mathbb{N}\mid 0\leq k\leq n-1, f(x)\in
[k/n,(k+1)/n]\textnormal{ for some }x\in M\}.\] For $n\geq 1$ and
$k\in K_n$ choose some $x_n^k\in M$ such that
$f(x_n^k)\in(k/n,(k+1)/n]$ and let
\[A_n^k=B(\succcurlyeq,x_n^k)\Big\backslash\bigcup_{k'>k}B(\succcurlyeq,x_n^{k'}).\]
Now let \[f_n=\sum_{k\in K_n}k/n 1_{A_n^k}\] Then
$|f(x)-f_n(x)|\leq 1/n$ for all $x\in M$, and the result follows.
\end{proof}

\begin{proof}[Proof of Theorem \ref{ParetoCoupling}] 
Define a relation $\succcurlyeq$ on $\mathcal{P}\times X$ such that
$(\succ,x)\succcurlyeq(\succ',x')$ if and only if $\succ=\succ'$ and
$x\succeq x'$. It is readily verified that $\succcurlyeq$ is a
preorder with a closed graph. Let $f:\mathcal{P}\times X\to\mathbb{R}$
be a $\succcurlyeq$-nondecreasing measurable function. We ca assume in
the following without loss of generality that the range of $f$ is
included in $[0,1]$. For each $\succ\in\mathcal{P}$, the section
$f_\succ:X\to [0,1]$ is $\succeq$-nondecreasing and measurable. So
using $\mu\succeq_D\mu'$, Lemma \ref{monapprox}, the dominated
convergence theorem, and Fubini's theorem for regular conditional
probabilities, we get
\[\int f~\mathrm d\nu=\int\int f_\succ~\mathrm d
g_\mu(\succ)~\mathrm d\mu_\mathcal{P}\ge\int\int f_\succ~\mathrm d
g_{\mu'}(\succ)~\mathrm d\mu_\mathcal{P}=\int f~\mathrm d\mu'.\] We
can, therefore, apply the pre-order version of Strassen's theorem in
\citet{MR1711599} to obtain a probability measure $\tau$ on
$\big(\mathcal{P}\times X\big)\times \big(\mathcal{P}\times X'\big)$
that is supported on the graph of $\succcurlyeq$ and has marginals
$\alpha$ and $\alpha'$, respectively. One obtains $\lambda$ from
$\tau$ by taking the appropriate marginal.
\end{proof}

\begin{lemma}\label{saabscont}Let $(T,\Sigma,\nu)$ be a superatomless
probability space and let $\mu$ be a probability measure on
$(T,\Sigma)$ that is absolutely continuous with respect to $\nu$. Then
$\mu$ is superatomless too.
\end{lemma}
\begin{proof}Let $h:T\to\mathbb{R}$ be a Radon-Nikodym derivative of
$\mu$ with respect to $\nu$ and let $E=\{t\in T\mid h(t)\neq
0\}$. Clearly, $\nu(E)>0$. Define $\phi:L_1(\nu,E)\to L_1(\mu)$ by
$\phi(f)(t)=fh^{-1}(t)$ for $t\in E$ and $\phi(f)(t)=0$ for
$t\notin E$. The operator $\phi$ is a linear isometry, so $L_1(\mu)$
is separable if and only if $L_1(E,\nu)$ is. The same construction can
be applied to subsets, showing that $\mu$ is superatomless.
\end{proof}

\begin{proof}[Proof of Theorem \ref{coalitionexist}]Let $h$ be a
Radon-Nikodym derivative of the $\mathcal{P}\times E\times X$-marginal
of $\mu$ with respect to $\nu\circ(c,\eta,f)^{-1}$. Without loss of
generality, we can take $h$ to have values in the unit interval. Let
$\lambda$ be the measure on $(T,\Sigma)$ with Radon-Nikodym derivative
$h\circ(c,\eta,f)$ with respect to $\nu$. Since $\lambda$ is
absolutely continuous with respect to a super-atomless measure,
$\lambda$ is superatomless too by Lemma \ref{saabscont}. Then
$\lambda\circ(c,\eta,f)^{-1}$ equals the
$\mathcal{P}\times E\times X$-marginal of $\mu$. By the saturation
property of superatomless measure spaces, there exists $g:T\to X'$
such that $\mu=\lambda\circ(c,\eta,f,g)$. Let
$\Sigma'=\sigma(c,\eta,f,g)$. Since $\nu$ is superatomless, there
exists $C\in\Sigma$ such that $\lambda(A)=\nu(A\cap C)$ for all
$A\in\Sigma'$. By the usual machinery,
$\int j~\textnormal{d}\lambda=\int_C j~\textnormal{d}\nu$ for $j$ a
nonnegative $\Sigma$-measurable real function. Applying this
coordinatewise, we get
\[\int_C g~\textnormal{d}\nu=\int g~\textnormal{d}\lambda=\int
\pi_{X'}~\textnormal{d}\mu=\int \pi_{E}~\textnormal{d}\mu=\int
\eta~\textnormal{d}\lambda=\int_C \eta~\textnormal{d}\mu.\] Since $c$
is $\Sigma'$-measurable, the set of $t$ such that $g(t)\succ_t f(t)$
lies in $\Sigma'$. Since
\[0=\mu\Big(\big\{(\succ,e,x,x')\mid x'\nsucc x\big\}\Big)=\nu\Big(C\cap
\big\{(\succ,e,x,x')\mid x'\nsucc x\big\}\Big),\] we have
$g(t)\succ_t f(t)$ for $\nu$-almost all $t\in C$. Finally, we define
$f'$ such that $f'$ coincides with $g$ on $C$ and $\eta$ outside
$C$. It is straightforward that $f'$ is an allocation with the desired
properties.
\end{proof}

\begin{proof}[Proof of Theorem \ref{approxcoalition}]
Write $\lambda$ for the marginal measure of $\mu$ on
$\mathcal{P}\times E\times X$ and let
$h:\mathcal{P}\times E\times X\to [0,1]$ be a Radon-Nikodym
derivative of $\lambda$ with respect to $\nu\circ(c,\eta,f)^{-1}$. Let
$\kappa_\mu:\mathcal{P}\times E\times X\to \Delta(X')$ be a
disintegration of $\mu$. Let $X_{\o}'=X'\cup\{\o\}$ with $\o\notin X$
and endow it with the Polish topology that makes $\o$ an isolated
point and $X'$ have the original topology as a subspace. Now define
$\kappa:T\to\Delta(X_{\o}')$ by

\[\kappa(t)=h\big(\succ(t),e(t),f(t)\big)\kappa_\mu\big(\succ(t),e(t),f(t)\big)+\Big(1-h\big(\succ(t),e(t),f(t)\big)\Big)\delta_{\o}.\]

Let $g_0:T\times X\to\mathbb{R}$ be given by $g_0(\o)=0$ and
$g_0(x)=1$ otherwise. Let $g_1:T\times X_{\o}'\to\mathbb{R}^l$ be
given by $g_1(t,\o)=0$ and $g_1(t)=e(t)-x'$ for $x'\neq\o$ and
$g_2:T\times X_{\o}\to\mathbb{R}$ be given by $g_2(t,x')=0$ if $x'=\o$
or $x'\succ(t) f(t)$ and $1$ otherwise. The integrand $g_0$ evaluated
at $\kappa$ is $\mu(\mathcal{P}\times E\times X\times X')$, and
the integrands $g_1$ and $g_2$ evaluated at $\kappa$ are $0$. By
Lyapunov's theorem for Young measures, \citet[Theorem
5.10]{MR1798830}, there is a function $f':T\to X'_{\o}$ such that
\[\int g_0\big(t,f'(t)\big)~\mathrm d\nu=\int \int g_0(t,x)~\mathrm
d\kappa(t)~\mathrm d\nu=\mu(\mathcal{P}\times E\times X\times X')>0,\]
\[\int g_1\big(t,f'(t)\big)~\mathrm d\nu=\int \int g_1(t,x)~\mathrm
d\kappa(t)~\mathrm d\nu=0,\] and
\[\int g_2\big(t,f'(t)\big)~\mathrm d\nu=\int \int g_2(t,x)~\mathrm
d\kappa(t)~\mathrm d\nu=0.\] Let $C=f'^{-1}(X)$. The first integral
equality implies that $\nu(C)>0,$ the second integral equality implies
that
\[\int_C f'(t)\mathrm d\nu=\int_C e(t)\mathrm d\nu,\]
and the third integral equality implies that $f'(t)\succ(t) f(t)$ for
almost all $t\in C$. It follows that $C$ blocks $f$ with the
individualistic allocation that assigns $f'(t)$ to each $t\in C$ and
$e(t)$ to each $t\notin C$.\end{proof}

\begin{proof}[Proof of Theorem \ref{Corecoupling}] Most of the proof
is analogous to the proof of Theorem \ref{ParetoCoupling}. What is new
is how we construct the marginal on $X''$. Suppose we have already
constructed the
$\mathcal{P}_{\textnormal{mo}}^*\times E\times X\times X'$-marginal of
$\lambda$. Call it $\tau$. Then we must have
\[\tau\Big(\big\{(\succ,e,x,x')\mid x'\succ x\big\}\Big)>0,\]
for otherwise we would get the converse ordering. For some $i$ with
$1\leq i\leq l$, the $i^\text{th}$ coordinate of
$\int\pi_E~\textnormal{d}\lambda$ must be positive. Moreover, there
must be some positive $n$ such that
\[\epsilon_n=\tau\Big(\big\{(\succ,e,x,x')\mid x'-1/n e_i\succ
x\big\}\Big)>0,\] where $e_i$ is the $i^\text{th}$ unit vector in
$\mathbb{R}^l$ and $x'-1/n e_i\succ x$ is taken to imply that
$x'-1/n e_i\ge 0$ so that the expression is well defined. Fix such an
$n$. If
\[\beta=\tau\Big(\big\{(\succ,e,x,x')\mid x'-1/n e_i\preceq
x\big\}\Big)=0,\] we are done and we can just take
$x''=x'$. Otherwise, define a transition probability
$\kappa:\mathcal{P}_{\textnormal{mo}}^*\times E\times X\times
X\to\Delta(X'')$ by $\kappa(\succ,e,x,x')=\delta_{x'-1/n e_i}$ if
$x'-e_i\succ x$ and $\kappa(\succ,e,x,x')=x'+\epsilon_n/\delta
e_i$. If we construct $\lambda$ from $\tau$ using $\kappa$, we get the
desired conclusion.
\end{proof}

\begin{proof}[Proof of Theorem \ref{coreequiv}] We prove the
nontrivial direction. Let $\nu$ be a core reallocation and let
\[N=\bigg\{z\in\mathbb{Q}^l\mid \nu\Big(\big\{(\succ,e,x)\mid  z+e\geq 0,
z+e\succ x\big\}\Big)>0\bigg\}.\] By monotonicity,
$N\neq\emptyset$. We show that
$\con(N)\cap\Int(-\mathbb{R}^l_+)=\emptyset$. Suppose for the sake of
contradiction that there are $z_1,\ldots, z_K\in N$ and
$\lambda_1,\ldots,\lambda_K\in\mathbb{R}_+$ such that
$\sum_{i=1}^K \lambda_i=1$ and $z=\sum_{i=1}^K \lambda_i z_i\ll
0$. For $i=1,\ldots,K$,
let\[E_i=\big\{(\succ,e,x)\mid z_i+e\succ x\big\}.\] Since preferences are
monotone and $z\ll 0$,
\[\mathcal{P}\times E\times
X=\bigcup_{n=1}^\infty\big\{(\succ,e,x)\mid e-nz\succ x\big\},\] which
implies
\[\lim_{n\to\infty}\nu\bigg(\big\{(\succ,e,x)\mid e-nz\succ
x\big\}\bigg)=\nu(\mathcal{P}\times E\times X),\] so for some
$N$, \[\nu\Big(\big\{(\succ,e,x)\mid e-Nz\succ x\big\}\Big)>0.\] Fix such
an $N$ and let \[E_0=\big\{(\succ,e,x)\mid e-Nz\succ x\big\}.\] For
$i=0,\ldots,K$ let $\nu_i$ be the measure that has Radon-Nikodym
derivative $1_{E_i}$ with respect to $\nu$ and let
$\alpha_i=\nu_i(\mathcal{P}\times E\times X)$. For $\rho>0$ small
enough, the measure
\[\rho(\alpha_0 N)^{-1}
\nu_0+\rho\sum_{i=1}^K\alpha_i^{-1}\lambda_i\nu_i\] is setwise smaller
than $\nu$.

Let
$\psi_i:\mathcal{P}\times E\times X\to\mathcal{P}\times E\times
X\times X'$ be given by \[\psi_0(\succ,e,x)=(\succ,e,x,e-Nz)\] and for
$i=1,\ldots,K$, define
$\psi_i:\mathcal{P}\times E\times X\to\mathcal{P}\times E\times
X\times X'$ by \[\psi_i(\succ,e,x)=(\succ,e,x,e+z_i).\]

Then
\[\kappa=\rho(\alpha_0 N)^{-1}
\nu_0\circ\psi_0^{-1}+\rho\sum_{i=1}^K\alpha_i^{-1}\lambda_i\nu_i\circ\psi_i^{-1}\]
witnesses to $\nu$ being blocked. This contradiction shows that
$\con(N)\cap\Int(-\mathbb{R}^l_+)=\emptyset$.

By the separation theorem, there exists $p\neq 0$ such that
$px\geq py$ for all $x\in \con(N)$ and $y\in
\Int(-\mathbb{R}^l_+)$. Clearly, $p\geq 0$. We have
\[\nu\bigg(\Big\{(\succ, e,x)\mid e+z\succ x\textnormal{ with
}z\in\mathbb{Q}^l\textnormal{ implies } p(e+z)\geq pe\Big\}\bigg)=1\]
and, since preferences are continuous, if $y\succ x$, there is a
sequence $\langle z_n\rangle$ in $\mathbb{Q}_l$ such that
$y=\lim_{n\to\infty}e+z_n$ and $e+z_n\succ x$ for all $n$. Therefore,
\[\nu\bigg(\Big\{(\succ, e,x)\mid y\succ x\textnormal{ implies } py\geq
pe\Big\}\bigg)=1.\] Since preferences are monotone and every point in
$X$ is the limit of a sequence of larger points,
$\nu\Big(\big\{(\succ, e,x)\mid px\geq pe\big\}\Big)=1$. If it would be
the case that $\nu\Big(\big\{(\succ, e,x)\mid px> pe\big\}\Big)>0$, then
\[p\int\pi_E~\textnormal{d}\nu=\int p\pi_E~\textnormal{d}\nu>\int
p\pi_X~\textnormal{d}\nu=p\int\pi_X~\textnormal{d}\nu,\] in
contradiction to
$\int\pi_E~\textnormal{d}\nu=\int\pi_X~\textnormal{d}\nu$. So
$\nu\Big(\big\{(\succ, e,x)\mid px=pe\big\}\Big)=1$. Since $p\geq 0$ and
$p\neq 0$ and $\int\pi_E~d\mu\gg 0$, we have
$\nu\Big(\big\{(\succ, e,x)\mid pe>0\big\}\Big)>0$. The usual cheaper
point argument implies then that $py> pe$ whenever $y\succ x$. By
strict monotonicity, this is only possible if $p\gg 0$, which is
therefore the case. But then $pe$ is not positive only when $e=0$, in
which case we must have $x=0$ and $y\succ x=0$ implies then that
$py>0=pe$. So $px\leq py$ and $y\succ x$ implies $py>pe$ holds for
$\nu$-almost all $(\succ,e,x)$ and $\nu$ is a Walrasian reallocation.
\end{proof}


\end{document}